\title{Estimating population average treatment effects from experiments with noncompliance
}
\author{Kellie N. Ottoboni \hspace{10mm}
Jason V. Poulos\thanks{\textbf{Corresponding author: Jason V. Poulos:} Department of Statistical Science, Duke University and SAMSI, Durham, North Carolina, 27708, United States; Email: \href{mailto:jason.poulos@duke.edu}{\nolinkurl{jason.poulos@duke.edu}}
	\newline
\textbf{Kellie N. Ottoboni:} Department of Statistics, University of California, Berkeley, California, 94720, United States}
\vspace{15mm}} 

\date{ }

\documentclass[hidelinks,12pt]{article}

\usepackage{hyperref, url} 
\usepackage{graphicx,amsfonts,psfrag,layout,subcaption,array,longtable,lscape,booktabs,dcolumn,amsmath,amssymb,amssymb,amsthm,setspace,epigraph,chronology,color,xcolor,colortbl,wasysym,diagbox,threeparttable}
\usepackage[]{graphicx}\usepackage[]{color}
\usepackage{enumitem}
\usepackage[page]{appendix}
\usepackage[section]{placeins}
\usepackage[linewidth=1pt]{mdframed}
\usepackage[margin=1in]{geometry} 

\usepackage{cite}

\usepackage[bottom]{footmisc}

\usepackage{tabularx}
\newcolumntype{Y}{>{\raggedleft\arraybackslash}X}

\newtheorem{theorem}{Theorem}

\newtheorem*{assumption*}{\assumptionnumber}
\providecommand{\assumptionnumber}{}
\makeatletter
\newenvironment{assumption}[2]
 {%
  \renewcommand{\assumptionnumber}{Assumption #1}%
  \begin{assumption*}%
  \protected@edef\@currentlabel{#1}%
 }
 {%
  \end{assumption*}
 }
\makeatother

\newcommand{\ind}{\mathbb{I}} 
\newcommand{\pr}{\mathbb{P}} 
\newcommand{\ex}{\mathbb{E}} 

\newcommand{\cov}{\mathrm{Cov}}

\newcommand\independent{\protect\mathpalette{\protect\independenT}{\perp}}
\def\independenT#1#2{\mathrel{\rlap{$#1#2$}\mkern2mu{#1#2}}}

\begin{document}
\begin{singlespace} 
\maketitle  
\end{singlespace}

\thispagestyle{empty}
\begin{abstract}  
\noindent 
Randomized control trials (RCTs) are the gold standard for estimating causal effects, but often use samples that are non-representative of the actual population of interest. We propose a reweighting method for estimating population average treatment effects in settings with noncompliance. Simulations show the proposed compliance-adjusted population estimator outperforms its unadjusted counterpart when compliance is relatively low and can be predicted by observed covariates. We apply the method to evaluate the effect of Medicaid coverage on health care use for a target population of adults who may benefit from expansions to the Medicaid program. We draw RCT data from the Oregon Health Insurance Experiment, where less than one-third of those randomly selected to receive Medicaid benefits actually enrolled.
\end{abstract}	

\pagebreak
\setcounter{page}{1} 

\vspace{20mm}

\section{Introduction}\label{intro}
Randomized control trials (RCTs) are the gold standard for estimating the causal effect of a treatment. An RCT may give unbiased estimates of sample average treatment effects, but external validity is an issue when RCT participants are unrepresentative of the actual population of interest. For example, participants in an RCT in which individuals volunteer to sign up for health insurance may be in poorer health at baseline than the overall population. External validity is particularly relevant to policymakers who require information on how the treatment effect would generalize to the broader population. 

A new research frontier in causal inference focuses on developing methods for extrapolating RCT results to a population \cite{ImaKinStu08,stuart2011use,Hartman}. Existing approaches to this problem are based in settings where there is full compliance with treatment; however, noncompliance is a prevalent issue in RCTs. Noncompliance occurs when individuals who are assigned to the treatment group do not accept the treatment. For individuals assigned to control, we are unable to observe who would have complied had they been assigned treatment. Noncompliance biases the intention--to--treat (ITT) estimate of the effect of treatment assignment toward zero.

We propose a reweighting method for estimating complier--average causal effects for the target population from RCT data with noncompliance, and refer to this estimator as the Population Average Treatment Effect on Treated Compliers (PATT-C). We model compliance in the RCT in order to predict the likely compliers in the RCT control group. Assuming that the response surface is the same for compliers in the RCT and population members who received treatment, we then predict the response surface for all RCT compliers and use the predicted values from the response surface model to estimate the potential outcomes of population members who received treatment, given their covariates. 

Our approach for estimating PATT-C differs from previous reweighting methods because we only need to estimate the potential outcomes for RCT compliers and we cannot observe who in the control group would have complied had they been assigned treatment. In Stuart et al. \cite{stuart2011use}, for example, a propensity score model is used to predict participation in the RCT, given pretreatment covariates common to both the RCT and population data. Individuals in the RCT and population are then weighted according to the inverse of the estimated propensity score. Hartman et al. \cite{Hartman} propose a method of reweighting the responses of individuals in an RCT according to the covariate distribution of the population. Reweighting methods typically leverage exchangeability of potential outcomes between the covariate-adjusted treated and controls in the RCT. In our approach, the potential outcomes between the complier treated and complier controls are not exchangeable by design, since we need to assume we know the compliance model.

When estimating the average causal effect for compliers from an RCT, researchers typically scale the estimated ITT effect by the compliance rate, assuming that there is only single crossover from treatment to control.\footnote{Alternative methods for compliance-adjusting sample estimates include estimating sharp bounds to the ITT effect in the presence of noncompliance \cite{balke1997bounds,imai2013experimental}; adjustment for treatment noncompliance using principal stratification \cite{frangakis2002principal,frumento2012evaluating}; and maximum-likelihood and Bayesian inferential methods \cite{yau2001inference}.} When extrapolating RCT results to a population, one might simply reweight the ITT effect according to the covariate distribution of the population and then divide by the proportion of treated compliers in the population in order to yield a population average effect of treatment on treated compliers. However, we do not observe the population compliance rate, and it is likely to differ across subgroups based on pretreatment covariates. By explicitly modeling compliance, our approach allows researchers to decompose population estimates by covariate group, which is useful for policymakers in evaluating the efficacy of policy interventions for subgroups of interest in a population. 

We apply the PATT-C estimator to measure the effect of Medicaid coverage on health care use for a target population of adults who may benefit from government-backed expansions to the Medicaid program. We are particularly interested in measuring the effect of Medicaid on emergency room (ER) use because it is the main delivery system through which the uninsured receive health care, and the uninsured could potentially receive higher quality health care through primary care visits. An important policy question is whether Medicaid expansions will decrease ER utilization and increase primary care visits by the previously uninsured. We draw RCT data from a large-scale health insurance experiment where less than one-third of those randomly selected to receive Medicaid benefits actually enrolled.

The paper proceeds as follows: Section~\ref{estimator} presents the proposed estimator and the necessary assumptions for its identifiability; Section~\ref{estimation} describes the estimation procedure; Section~\ref{sim} reports the estimator's performance in simulations; Section~\ref{application} uses the estimator to identify the effect of extending Medicaid coverage to the low--income adult population in the U.S.; Section~\ref{discussion} discusses the results and offers direction for future research. 

\section{Estimator} \label{estimator} 
We are interested in using the results of an RCT to estimate complier--average causal effects for a target population. Compliance with treatment in the population is not assigned at random. It may depend on unobserved variables, thus confounding the effect of treatment received on the outcome of interest. RCTs are needed to isolate the effect of treatment received. 

Ideally, we would take the results of an RCT and reweight the sample such that the reweighted covariates match the those in the population. In practice, one rarely knows the true covariate distribution in the target population. Instead, we consider data from an observational study in which participants are representative of the target population. The proposed estimator combines RCT and observational data to overcome these issues.

\subsection{Assumptions} \label{assumptions}
Let $Y_{isd}$ be the potential outcome for individual $i$ in group $s$ and treatment received $d$. Let $S_i \in \{0,1\}$ denote the sample assignment, where $s=0$ is the population and $s=1$ is the RCT. Let $T_i \in \{0,1\}$ indicate treatment assignment and $D_i \in\{0,1\}$ indicate whether treatment was actually received. Treatment is assigned at random in the RCT, so we observe both $D_i$ and $T_i$ when $S_i = 1$. For compliers in the RCT, $D_i = T_i$.
	
The absence of $T_i$ in the subscript of the potential outcomes notation implicitly assumes the exclusion restriction for noncompliers, which precludes treatment assignment from having an effect on the potential outcomes of noncompliers in the RCT \cite{imbens2015causal}. Also implicit in our notation is the assumption of no interference between units, which prevents the potential outcomes for any unit from varying with treatments assigned to other units \cite{rubin1990}.

Let $W_i$ be individual $i$'s observable pretreatment covariates that are related to the sample selection mechanism for membership in the RCT, treatment assignment in the population, and complier status. Let $C_i$ be an indicator for individual $i$'s compliance with treatment, which is only observable for individuals in the RCT treatment group. In the population, we suppose that treatment is made available to individuals based on their covariates $W_i$.\footnote{We use the same $W_i$ across all identifying assumptions, which implicitly assumes that the observable covariates that determine sample selection also determine population treatment assignment and complier status.} Individuals with $T_i = 0$ do not receive treatment, while those with $T_i=1$ may decide whether or not to accept treatment. For individuals in the population, we only observe $D_i$ --- not $T_i$.

\vskip 0.2in
\begin{assumption}{1}{}\label{consistency}
	Consistency under parallel studies:
	\begin{equation*}
	Y_{i0d} = Y_{i1d} \qquad \forall \, i \, , \, d=\{0,1\}.
	\end{equation*}
\end{assumption} 

\noindent Assumption \eqref{consistency} requires that each individual $i$ has the same response to treatment received, whether $i$ is in the RCT or not. Compliance status $C_i$ is not a factor in this assumption because we assume that compliance is conditionally independent of sample and treatment assignment for all individuals with covariates $W_i$.
	
\vskip 0.2in
\begin{assumption}{2}{}\label{compl}
	Conditional independence of compliance and sample and treatment assignment:
	\begin{equation*}
	C_i \independent S_i,\,T_i \mid W_i, \qquad 0 < \pr(C_i = 1 \mid W_i) < 1. 
	\end{equation*}
\end{assumption}

\noindent Assumption \eqref{compl} implies that $P(C_i=1 | S_i=1, T_i=1, W_i) = P(C_i=1 | S_i=1, T_i=0, W_i)$, which is useful when predicting the probability of compliance as a function of covariates $W_i$. Together, Assumptions \eqref{consistency} and \eqref{compl} ensure that potential outcomes do not differ based on sample assignment or receipt of treatment.

Assumption \eqref{si_treat} ensures the potential outcomes for treatment are independent of sample assignment for individuals with the same covariates $W_i$ and treatment assignment.

\vskip 0.2in
\begin{assumption}{3}{}\label{si_treat}
	Strong ignorability of sample assignment for treated:
	\begin{equation*}
		(Y_{i01}, Y_{i11}) \independent S_i \mid (W_i, T_i=1,C_i=1), \qquad 0 < \pr(S_i=1 \mid W_i, T_i=1,C_i=1) <1.
	\end{equation*}
\end{assumption}
\noindent
We make a similar assumption for the potential outcomes under control.
\vskip 0.2in
\begin{assumption}{4}{}\label{si_ctrl}
	Strong ignorability of sample assignment for controls:
	\begin{equation*}
		(Y_{i00}, Y_{i10}) \independent S_i \mid (W_i, T_i=1,C_i= 1), \qquad 0 < \pr(S_i=1 \mid W_i, T_i=1,C_i=1) <1. 
\end{equation*}\end{assumption}
\noindent
Note that Assumptions \eqref{si_treat} and \eqref{si_ctrl} imply strong ignorability of sample assignment for treated and control noncompliers since Assumption \eqref{compl} states that compliance is also independent of sample and treatment assignment, conditional on $W_i$. However, we are interested only on modeling the response surfaces for compliers.

Restrictive exclusion criteria in RCTs can result in a sample covariate distribution that differs substantially from the population covariate distribution, thereby reducing the external validity of RCTs \cite{rothwell2005external}. High rates of exclusion pose a threat to Assumptions \eqref{si_treat} and \eqref{si_ctrl} if exclusion increases the likelihood that there are unobserved differences between the RCT and target population that are correlated with potential outcomes. For example, the RCT in our application required enrolled participants to recertify their eligibility status every six months during the study period. The exclusion of participants who failed to recertify because their household income exceeded a given cutoff threatens strong ignorability if the factors that contributed to the failure to recertify are correlated with unobservables that are also correlated with potential outcomes. While we cannot directly test the strong ignorability assumptions, bias arising from violations of these assumptions would cause the placebo tests described in Section~\ref{placebo-tests} to fail.

We include an additional assumption to identify the average causal effect for compliers, as in Angrist et al. \cite{Angrist1996}. 

\vskip 0.2in
\begin{assumption}{5}{}\label{monotonicity}
	One-sided noncompliance: 
	\begin{equation*}
	\pr(D_i \mid T_i = 0) = 0, \qquad \forall \, i.
	\end{equation*}
\end{assumption}
\noindent
Assumption \eqref{monotonicity} ensures that noncompliance is one-sided; i.e., individuals assigned to control are not allowed to receive treatment. It explicitly rules out the existence of individuals who always receive treatment (i.e, never-takers) and those who receive the opposite of their treatment assignment (i.e., defiers). 

\subsection{Causal diagram}\label{causal-diagram}

Figure~\ref{fig:DAG} shows the assumptions needed to identify PATT-C in the form of a causal diagram \cite{pearl1995causal}. The missing arrow between $T_i$ (or $S_i$) and $C_i$ signifies that treatment (or sample) assignment may only depend on compliance status through covariates $W_i$, by Assumption \eqref{compl}. Likewise, the missing arrow between $Y_{isd}$ and $T_i$ (or $S_i$) signifies that potential outcomes may only depend on treatment (or sample) assignment through covariates, by Assumptions \eqref{si_treat} and \eqref{si_ctrl}. 

Confounding arcs represent potential back-door paths that contain only unobserved variables. The existence of back-door paths from $D_i$ to $Y_{isd}$ through $S_i$ and $C_i$ imply that the average causal effect of $D_i$ on $Y_{isd}$ cannot be identified by just conditioning on $W_i$; that is, we cannot ignore the role of $S_i$ and $C_i$. From the internal validity standpoint, the role of $W_i$ is critical: in the presence of unobserved confounders, there is a back-door pathway from $T_i$ back to $W_i$ and into $Y_{isd}$.

\begin{figure}[htb]
	\begin{center}
		\includegraphics[width = 0.8\textwidth]{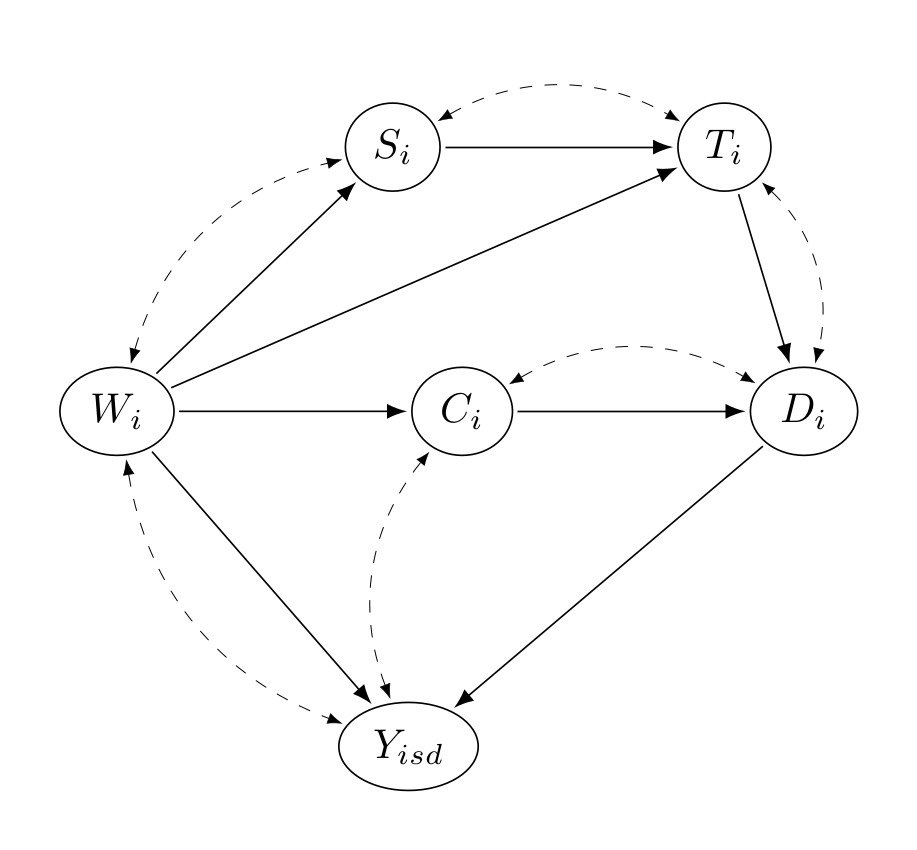}
		\end{center}
		\caption{Causal diagram representing the effect of treatment received, $D_i$, on potential outcomes, $Y_{isd}$. Solid unidirectional arrows represent causal links between observed variables. Dashed bidirectional arrows represent confounding arcs containing only latent variables, which are implicit in the confounding arcs.\label{fig:DAG}}
	\end{figure}

\subsection{PATT-C}\label{pattc-sub}

PATT-C is the average causal effect of taking up treatment estimated on individuals who received treatment in the population. This interpretation follows directly from Assumptions \eqref{consistency} and \eqref{compl}, which ensure that potential outcomes do not differ based on sample assignment or receipt of treatment. It is written as follows:
\begin{equation}\label{tpattc}
\tau_{\text{PATT-C}} = \ex\left( Y_{i01} - Y_{i00} \mid S_i=0, D_i=1\right).
\end{equation}
Theorem \eqref{thm1} relates the treatment effect in the RCT to the treatment effect in the population (proof given in Appendix~\ref{appendix-proof}).

\vskip 0.2in
\begin{theorem}\label{thm1}
Under Assumptions \eqref{consistency} -- \eqref{monotonicity},
\begin{equation}\label{tpattc-est}
\tau_{\text{PATT-C}} = \ex_{01}\left[  \ex\left(Y_{i11} \mid S_i=1, D_i=1, W_i\right)\right]
-\ex_{01}\left[  \ex\left(Y_{i10} \mid S_i=1, T_i =0, C_i =1, W_i\right) \right], 
\end{equation}
where $\ex_{01}\left[\ex(\cdot \mid\dots, W_i)\right]$ denotes the expectation with respect to the distribution of $W_i$ for population members who received treatment. 
\end{theorem}

\section{Estimation procedure}\label{estimation}

There are two challenges in turning Theorem~\eqref{thm1} into an estimator of $\tau_{\text{PATT-C}}$ in practice. First, we must estimate the inner expectation over potential outcomes of compliers in the RCT. In the empirical example, we use an ensemble of algorithms to estimate the response surface for compliers in the RCT, given their covariates. Thus, the first term in the expression for $\tau_{\text{PATT-C}}$ is estimated by the weighted average of points on the response surface, evaluated for each treated population member's potential outcome under treatment. The second term is estimated by the weighted average of points on the response surface, evaluated for each treated population member's potential outcome under control.

The second challenge is that we cannot observe which individuals are included in the estimation of the second term. In the RCT control group, $C_i$ is unobservable, as they always receive no treatment ($D_i=0$). We must estimate the second term of Eq.~\eqref{tpattc-est} by predicting who in the control group would be a complier had they been assigned to treatment. Explicitly modeling compliance allows us to decompose PATT-C estimates by subgroup according to covariates common to both RCT and observational datasets. This approach also accounts for settings where the compliance rate differs between the sample and population, as well as across subgroups.

The procedure for estimating $\tau_{\text{PATT-C}}$ using Theorem~\eqref{thm1} is as follows:
\begin{enumerate}[label=\textbf{S.\arabic*},ref=S.\arabic*]
\item Using the group assigned to treatment in the RCT $(S_i=1, T_i=1)$, train a model (or an ensemble of models) to predict the probability of compliance as a function of covariates $W_i$. \label{compliance-model}
\item Using the model from \ref{compliance-model}, predict who in the RCT assigned to control \textit{would have} complied to treatment had they been assigned to the treatment group.\footnote{We dichotomize the predicted values based on an ``optimal'' cut-point that minimizes the Euclidean distance between the ROC curve and the (0,1) corner of the ROC plane. This method, which is commonly used in the epidemiology literature \cite{perkins2006inconsistency}, minimizes misclassification error by finding the cut-point closest to where the true positive rate (i.e., the proportion of sample compliers correctly identified) is 1 and the inverse of the true negative rate (i.e., the share of noncompliers correctly identified) is 0. We locate a prediction threshold of 45\% that balances the true positive and true negative rates at about 70\% each.}  \label{complier-predict}
\item For the observed compliers assigned to treatment and predicted compliers assigned to control, train a model to predict the response using $W_i$ and $D_i$, which gives $\ex(Y_{i1d} \mid S_i=1, D_i=d, W_i)$ for $d \in \{0,1\}$.\label{response-model}
\item For all individuals who received treatment in the population $(S_i=0, D_i=1)$, estimate their potential outcomes using the model from \ref{response-model}, which gives $Y_{i1d}$ for $d \in \{0,1\}$. The mean counterfactual $Y_{i11}$ minus the mean counterfactual $Y_{i10}$ is the estimate of $\tau_{\text{PATT-C}}$.\label{response-predict}
\end{enumerate}

Observe that Assumptions \eqref{si_treat} and \eqref{si_ctrl} are particularly important for estimating $\tau_{\text{PATT-C}}$: the success of the proposed estimator hinges on the assumption that the response surface is the same for compliers in the RCT and target population. If this does not hold, then the potential outcomes $Y_{i10}$ and $Y_{i11}$ for target population individuals cannot be estimated using the model from \ref{response-model}. Section~\ref{verifying} discusses whether the strong ignorability assumptions are plausible in the empirical application.

\subsection{Modeling assumptions}  \label{modeling-assumptions}

In addition to the identification assumptions, we require additional modeling assumptions for the estimation procedure. As pointed out in Section~\ref{assumptions}, we require that $W_i$ is complete because if any relevant elements of $W_i$ are not controlled, then there is a backdoor pathway from $T_i$ back to $W_i$ and into $Y_{isd}$. Additionally, we assume that the compliance model is accurate in predicting compliance in the training sample of RCT participants assigned to treatment and also generalizable to RCT participants assigned to control (\ref{compliance-model} and \ref{complier-predict}). We describe below the method of evaluating the generalizability of the compliance model.

\subsection{Ensemble method}  \label{ensemble}

In the empirical application, we use the super learner weighted ensemble method \cite{van2007super,polley2011super} for the estimation steps \ref{compliance-model} and \ref{response-model}. The super learner combines algorithms with a convex combination of weights based on minimizing cross-validated error, and typically outperforms single algorithms selected by cross-validation. 

We choose a variety of candidate algorithms to construct the ensemble, with a preference for algorithms that tend to perform well in supervised classification tasks. We also have a preference for algorithms that have a built-in variable selection property. The idea is that we input the same $W_i$ and each candidate algorithm selects the most important covariates for predicting compliance status or potential outcomes.\footnote{A potential concern when predicting potential outcomes is that the algorithm might shrink the treatment received predictor to zero, which would result in no difference between counterfactual potential outcomes.} We select three types of candidate algorithms: additive regression models \cite{buja1989linear}; gradient boosted regression \cite{friedman2001greedy}; L1 or L2-regularized linear models (i.e., Lasso or ridge regression, respectively) \cite{tibshirani2012strong}; and ensembles of decision trees (i.e., random forests) \cite{breiman2001}. Lasso is particularly attractive because it tends to shrink all but one of the coefficients of correlated covariates to zero. 

\section{Simulations} \label{sim}

We conduct a simulation study comparing the performance of the PATT-C estimator against its unadjusted analogue, which we refer to as the Population Average Treatment Effect on the Treated (PATT): 
\begin{equation}\label{tpatt-est}
\tau_{\text{PATT}} = \ex_{01}\left[  \ex\left(Y_{i11} \mid S_i=1, T_i =1, W_i\right)\right]
-\ex_{01}\left[  \ex\left(Y_{i10} \mid S_i=1, T_i =1, W_i\right) \right].
\end{equation}
\noindent 
Eq. \eqref{tpatt-est} identifies the population-average causal effect of treatment assignment, adjusted according to the covariate distribution of population members who received treatment. PATT is estimated following the same estimation procedure as PATT-C, except that in estimation step~\ref{response-model} the response curve is estimated on all RCT participants, regardless of compliance status, conditional on their covariates and actual treatment received. Identical to~\ref{response-predict} in the estimation procedure for PATT-C, we then use the response model to estimate the outcomes of population members who received treatment given their covariates. Like the PATT-C estimator, the PATT estimator crucially relies on the assumption that the response surface is the same for RCT participants and population members who received treatment.

We compare the population estimators against the sample Complier Average Causal Effect (CACE) \cite{imbens1997bayesian}, which is commonly referred to the Local Average Treatment Effect (LATE) in the econometrics literature \cite{imbens1994identification,Angrist1996}. In the context of program evaluation, it is more relevant than the ITT estimator because only RCT participants who received treatment would have their outcomes affected by treatment in the presence of a nonnegative treatment effect.

CACE is defined as the average causal effect of treatment received restricted to sample compliers:
\begin{equation}\label{tcace}
\tau_{\text{CACE}} = \ex\left( Y_{i11} - Y_{i10} \mid S_i=1, C_i=1\right).
\end{equation}	
In other words, CACE is the treatment effect for RCT participants who would comply regardless of treatment assignment. However, we are unable to observe the compliance status of RCT participants assigned to control because we do not know if they would have complied if they had been assigned to treatment. A generalization of the instrumental variables estimator of the CACE in the presence of noncompliers is given by:
\begin{equation}\label{tcace-hat}
\hat{\tau}_{\text{CACE}} = \frac{\ex\left( Y_{i11} - Y_{i10} \mid S_i=1\right)}{\pr(T_i = D_i = 1 \mid S_i=1)},
\end{equation}	
which is identified under Assumption \eqref{monotonicity}. The estimator is equivalent to scaling the ITT effect by the sample proportion of treated compliers \cite{freedman2006}.

\subsection{Simulation design}

The simulation is designed so that the effect of treatment is heterogeneous and depends on covariates which are different in the RCT and target population. The design satisfies the conditional independence assumptions in Figure~\ref{fig:DAG}.

In the simulation, RCT eligibility, complier status, and treatment assignment in the population depend on multivariate normal covariates $(W^{1}_i, W^{2}_i, W^{3}_i, W^{4}_i)$ with means $(0.5, 1, -1, -1)$ and covariances $\cov(W^{1}_i, W^{2}_i) = \cov(W^{1}_i, W^{4}_i)= \cov(W^{2}_i, W^{4}_i) = \cov(W^{3}_i, W^{4}_i) = 1$ and $\cov(W^{1}_i, W^{3}_i) = \cov(W^{2}_i, W^{3}_i) = 0.5$.  The first three covariates are observed by the researcher and $W^{4}_i$ is unobserved. $U_i, V_i, R_i$, and $Q_i$ are standard normal error terms. $U_i, V_i, R_i, Q_i$, and $(W^{1}_i, W^{2}_i, W^{3}_i, W^{4}_i)$ are mutually independent.

The equation for selection into the RCT is
\begin{equation*}
S_i = \ind(e_2 + g_1W^{1}_i + g_2W^{2}_i + g_3W^{3}_i + e_4W^{4}_i + R_i > 0).
\end{equation*}	
The parameter $e_2$ varies the fraction of the population eligible for the RCT and $e_4$ varies the degree of confounding with sample selection. We set the constants $g_1, g_2,$ and $g_3$ to be $0.5,$ $0.25,$ and $0.75$, respectively. 

Complier status is determined by
\begin{equation*}
C_i = \ind(e_3 + h_2W^{2}_i + h_3W^{3}_i + e_5W^{4}_i + Q_i > 0),
\end{equation*}	
where $e_3$ varies the fraction of compliers in the population, and $e_5$ varies the degree of confounding with treatment assignment. We set the constants $h_2$ and $h_3$ to $0.5$. 

For individuals in the population ($S_i=0$), treatment is assigned by
\begin{equation*}
T_i = \ind(e_1 + f_1W^{1}_i + f_2W^{2}_i + e_6W^{4}_i + V_i > 0),
\end{equation*}	
where $e_1$ varies the fraction eligible for treatment in the population and $e_6$ varies the degree of confounding with sample selection. We set the constants $f_1$ and $f_2$ to $0.25$ and $0.75$, respectively. For individuals in the RCT ($S_1=1$), treatment assignment $T_i$ is a sample from a Bernoulli distribution with probability $p=0.5$.

Finally, the response is determined by 
\begin{align*}
Y_{isd} &= a + bD_i + c_1W^{1}_i + c_2W^{2}_i + c_3W^{4} +dU_i,
\end{align*}
where we set $a, c_1, c_3,$ and $d$ to $1$ and $c_2$ to $2$. The treatment effect $b$ is heterogeneous:
\begin{align*}
b &= \left\{
			\begin{array}{rr}
        1, & \text{if } W^{1}_i > 0.75\\
		-1, & \text{if } W^{1}_i \leq 0.75\\
			\end{array}\right.
\end{align*}	 
 
We generate a population of 30,000 individuals and randomly sample 5,000. Those among the 5,000 who are eligible for the RCT ($S_i=1$) are selected. Similarly, we sample 5,000 individuals from the population and select those who are not eligible for the RCT ($S_i=0$) to be our observational study participants.\footnote{This set-up mimics the reality that a population census is usually impossible.} We set each individual's treatment received $D_i$ according to their treatment assignment and complier status and observe their responses $Y_{isd}$. In this design, the manner in which $S_i$, $T_i$, $D_i$, $C_i$, and $Y_{isd}$ are simulated ensures that Assumptions \eqref{consistency} -- \eqref{monotonicity} hold.

In the assigned-treatment RCT group $(S_i = 1, T_i = 1)$, we train a gradient boosted regression on the covariates to predict who in the control group $(S_i = 1, T_i = 0)$ would comply with treatment ($C_i=1$), which is unobservable. These individuals \textit{would have} complied had they been assigned to the treatment group. For this group of observed compliers to treatment and predicted compliers from the control group of the RCT, we estimate the response surface by training another gradient boosted regression on features $(W^{1}_i, W^{2}_i, W^{3}_i)$ and $D_i$.

\subsection{Simulation results}\label{sim-results}

We vary each of the parameters $e_1$, $e_2$, $e_3$, $e_4$, $e_5$, and $e_6$ along a grid of five random standard normal values in order to generate different combinations of rates of compliance, treatment eligibility, RCT eligibility in the population, and confounding. For each possible combination of the six parameters, and holding all other parameters constant, we compute over 10 simulation runs the average root mean squared error (RMSE) between the true population average treatment effect and the PATT-C, PATT, or CACE estimates. Averaging across combinations, PATT yields the highest average RMSE (1.06), followed by the CACE (0.89), and the PATT-C (0.76). Since PATT does not adjust for compliance, we would expect bias relative to the PATT-C in settings with noncompliance.

Figures~\ref{fig:rmse_ratec_rates} and \ref{fig:rmse_ratec_ratet} show the average RMSE of the estimators as a function of the population compliance rate and the share of population members eligible to participate in the RCT or the population treatment rate, respectively. The PATT estimator does not correct for bias resulting from noncompliance in the population and consequently performs poorly when the population compliance rate is relatively low (i.e., $\leq 60\%$). The PATT-C estimator corrects for noncompliance in the population and thus outperforms the PATT in low-compliance settings. The CACE corrects for noncompliance in the sample, and underperforms compared to the PATT-C due to differences between the sample and population.

\begin{figure}[htbp]
	\begin{center}
		\includegraphics[width = \textwidth]{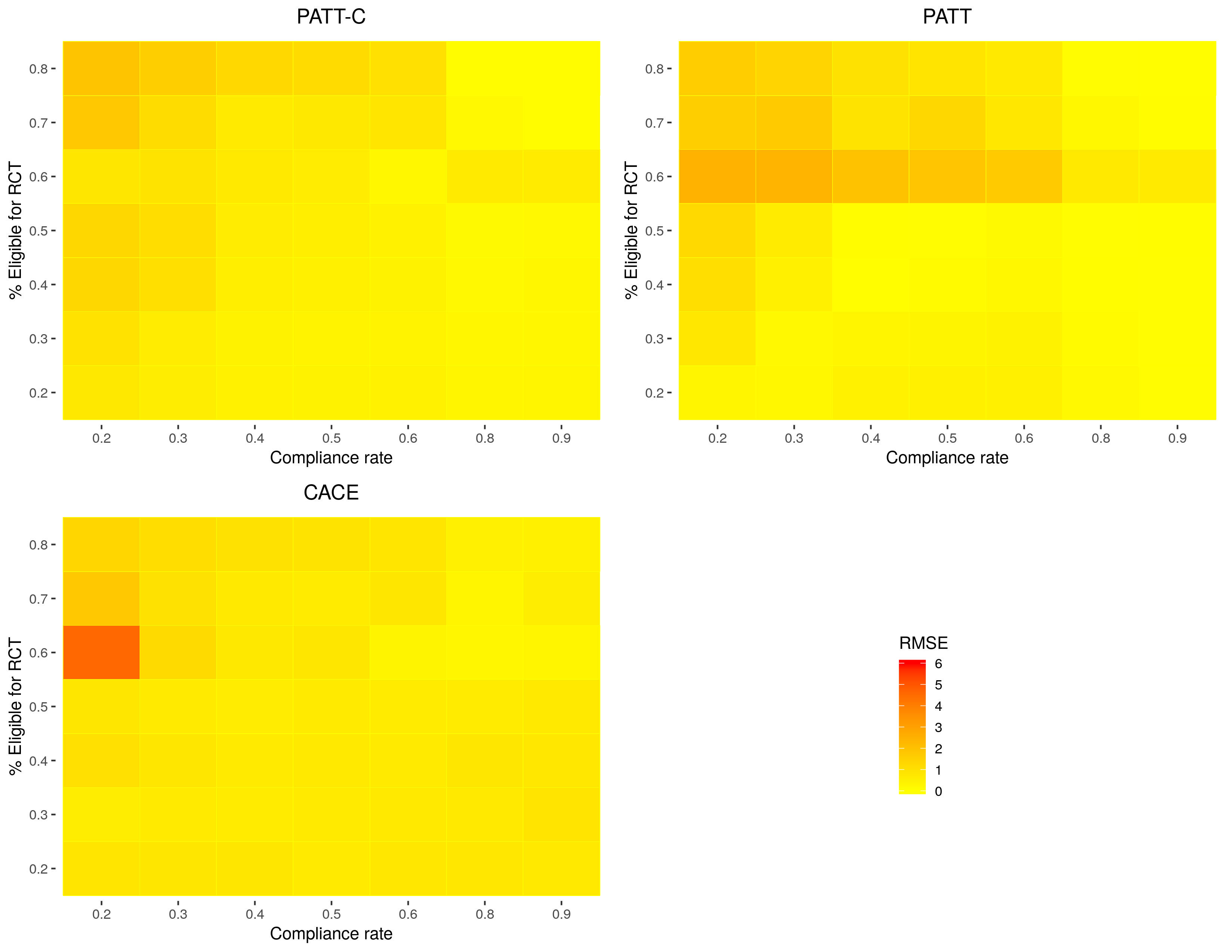}
		\caption{Average RMSE, binned by compliance rate and percent eligible for the RCT. \emph{Notes:} Darker tiles correspond to higher errors and white tiles correspond to missing simulated data.\label{fig:rmse_ratec_rates}}
	\end{center}
\end{figure}

\begin{figure}[htbp]
	\begin{center}
		\includegraphics[width = \textwidth]{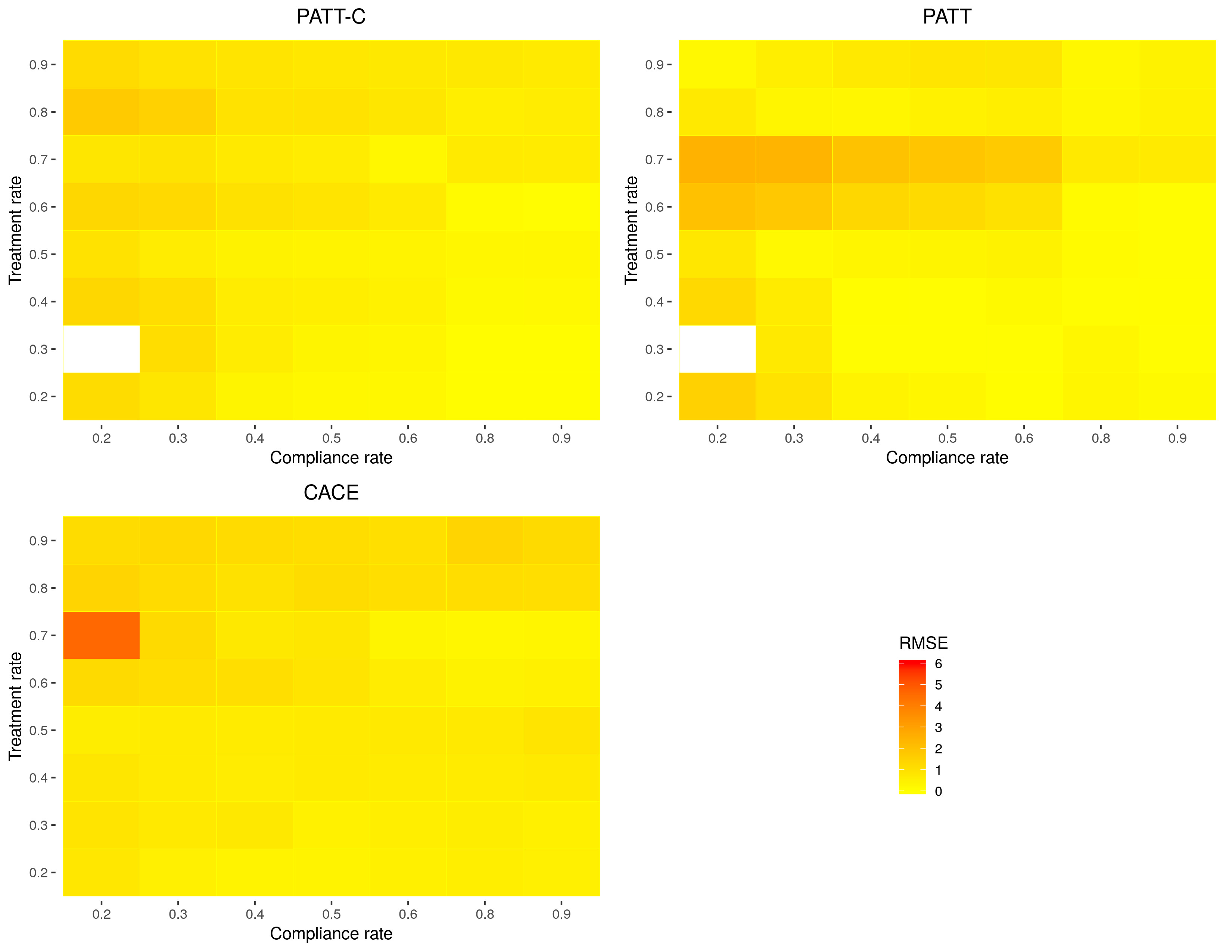}
		\caption{Average RMSE binned by compliance rate and treatment rate.\label{fig:rmse_ratec_ratet}}
	\end{center}
\end{figure}

Figure~\ref{fig:rmse_boxplots_rateC} compares the average RMSE of the estimators at varying levels of compliance in the population. The error for each of the estimators decreases as a greater share of population members comply with treatment. PATT-C outperforms both PATT and CACE in terms of minimizing RMSE when the population compliance rate is below 90\%. The PATT outperforms the PATT-C only at a 90\% compliance rate, and the CACE outperforms the PATT when the population compliance rate is at 60\% or below. 

\begin{figure}[htbp]
	\begin{center}
		\includegraphics[width = \textwidth]{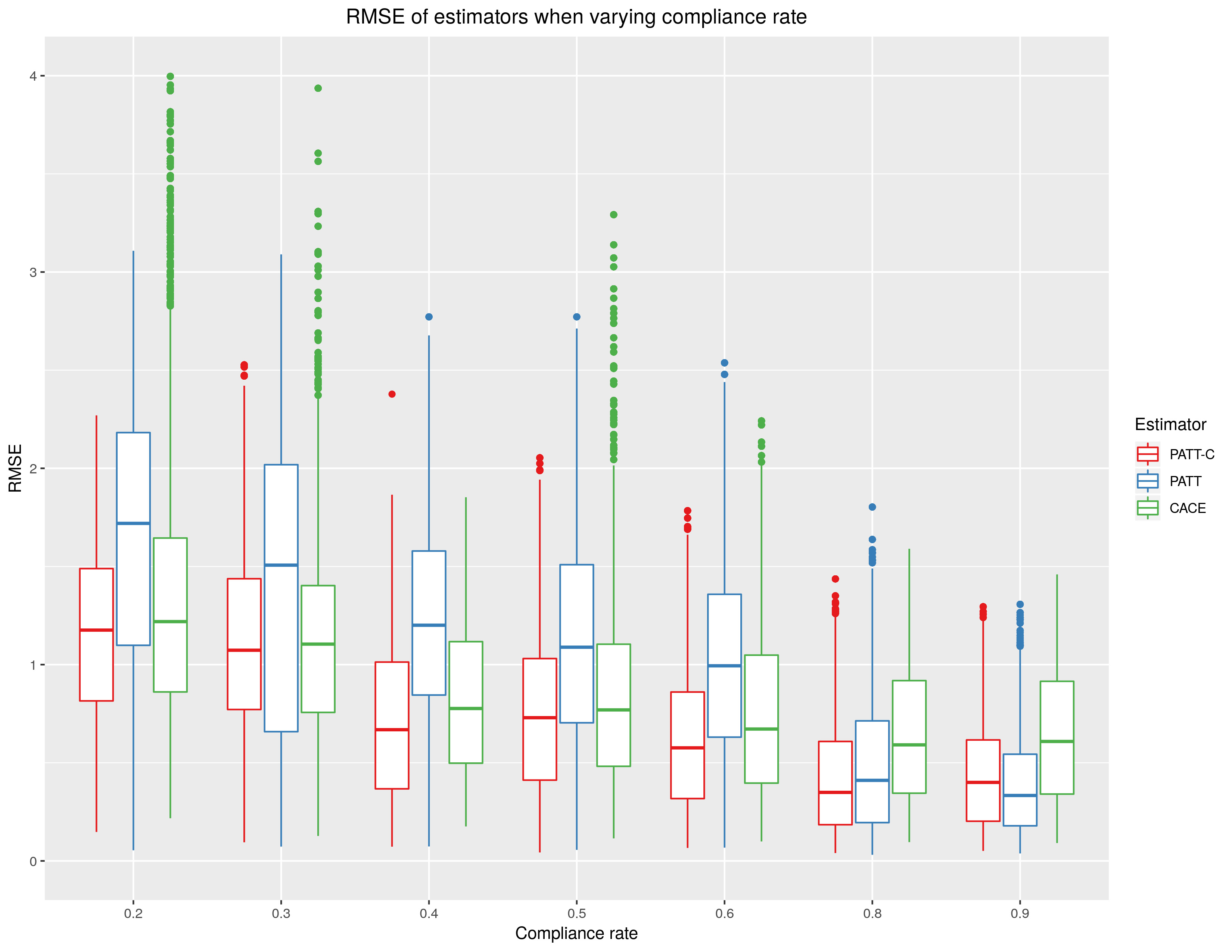}
		\caption{Average RMSE according to compliance rates in the population.\label{fig:rmse_boxplots_rateC}}
	\end{center}
\end{figure}

In the Appendix, Figures~\ref{fig:rmse_boxplots_RateConC}, \ref{fig:rmse_boxplots_RateConT}, and \ref{fig:rmse_boxplots_RateConS} plot the relationships between estimation error and the degrees of confounding in the mechanisms that determine compliance, treatment assignment, and sample selection, respectively. The estimation error of PATT-C is comparatively less invariant to increases in the degree of confounding in the three mechanisms compared to its unadjusted counterpart. The estimation error of CACE is generally more variable than that of the population estimators due to CACE's inability to account for differences between the sample and target population.

\pagebreak
\section{Application: Medicaid and health care use} \label{application}

We apply the PATT-C estimator to measure the effect of Medicaid coverage on health care use for a target population of adults who may benefit from expansions to the Medicaid program. In particular, we examine the population of nonelderly adults in the U.S. with household incomes at or below 138\% of the Federal Poverty Level (FPL) --- which amounts to \$32,913 for a four--person household in 2014 --- who may be eligible for Medicaid following the Affordable Care Act (ACA) Medicaid expansion.

We draw RCT data from the Oregon Health Insurance Experiment (OHIE) \cite{finkelstein2012,baicker2013,baicker2014,Taubman}, which randomly assigned Medicaid coverage to the uninsured and examined their subsequent health care use. Subsequent research calls in to question the external validity of the OHIE, which resulted in the counterintuitive finding that Medicaid increased ER use among RCT participants. Quasi-experimental studies on the impact of the 2006 Massachusetts health reform --- which served as a model for the ACA --- show that ER use decreased or remained constant following the reform \cite{miller2012effect, kolstad2012impact}. Kowalski \cite{NBERw22363} re-weights the LATE estimated on the OHIE sample to a broader population in Massachusetts using observational data from the Behavioral Risk Factor Surveillance System (BRFSS) \cite{BRFSS} and finds Medicaid significantly decreased the probability of visiting the ER, and had no significant effect on the number of ER visits. 
 
\subsection{RCT sample} 

In 2008, a group of uninsured low-income adults participated in the OHIE for the chance to apply to receive health insurance through a state Medicaid program. In line with program eligibility requirements, participants were restricted to Oregon residents aged 19 to 64 who were not otherwise eligible for public insurance, who had been without insurance for six months, had income below the FPL, and held assets below \$2,000. Treatment assignment occurred at the household level: participants selected by the lottery won the opportunity for themselves and any household member to apply for Medicaid. Within a sample of 74,922 individuals representing 66,385 households, 29,834 participants were selected by the lottery; the remaining 45,008 participants served as controls in the experiment. 

Participants in selected households were enrolled in Medicaid if they returned an enrollment application within 45 days of receipt. Only 30\% of participants in selected households successfully enrolled. The low compliance rate is primarily due to failure to return an application or demonstrate income below the FPL. Compliance is measured using a binary variable indicating whether the participant was enrolled in any Medicaid program during the study period.

We include as covariates in our response and complier models (\ref{compliance-model} and \ref{response-model}, respectively) pretreatment information on participant age, race, gender, education, marital status, number of children in the household, employment status, health status, and household income. We also include indicator variables on household size because lottery selection was random conditional on the number of household members. All analyses in the current application cluster-adjust standard errors at the household level because treatment occurs at the household level. The analyses also use survey weights to adjust for the probability of being sampled and non-response.

The response data originate from a mail survey containing questions about health insurance and health care use. The response variables measure health care use in terms of the number of ER and primary care (i.e., outpatient) visits in the past six months. Following Finkelstein et al. \cite{finkelstein2012}, indicator variables for survey wave and interactions with household size indicators are also included as predictors in the response and complier models because the proportion of treated participants varies across the survey waves.

\subsection{Observational data} 

We acquire data on the target population from the National Health Interview Study (NHIS) \cite{NHIS} for the period 2008 to 2017.\footnote{A possible limitation of this application is that it ignores the complex sampling techniques of the NHIS sample design such as differential sampling, which is discussed in detail in Parsons et al. \cite{parsons2014design}.} We restrict the sample to respondents with income below 138\% of the FPL and who are uninsured or on Medicaid and select covariates on respondent characteristics that match the OHIE pretreatment covariates. We use a recoded variable that indicates whether respondents are on Medicaid as an analogue to the OHIE compliance measure. The outcomes of interest from the NHIS are based on questions that are virtually identical to the OHIE mail survey questions, except that the utilization questions in the NHIS are asked with a 12 month rather than a 6 month look-back period. Following Finkelstein et al. \cite{finkelstein2012}, we resolve this discrepancy by halving the NHIS responses in order to make them comparable to the OHIE outcomes.

\subsection{Verifying assumptions} \label{verifying}

We verify the identification assumptions needed to identify $\tau_{\text{PATT-C}}$ prior to conducting placebo tests in Section~\ref{placebo-tests} and reporting the results in Section~\ref{results}.

\subsubsection{Consistency} Assumption \eqref{consistency} ensures that potential outcomes for participants in the target population would be identical to their outcomes in the RCT if they had been randomly assigned their observed treatment. In the empirical application, Medicaid coverage for uninsured individuals was applied in the same manner in the RCT as it is in the population. Differences in potential outcomes due to sample selection might arise, however, if there are differences in the mail surveys used to elicit health care use responses between the RCT and the nonrandomized study. 

\subsubsection{Conditional independence} Assumption \eqref{compl} is violated if assignment to treatment influences the compliance status of individuals with the same covariates. The compliance ensemble can accurately classify compliance status for 78\% of treated RCT participants with only the number of household members, survey wave (and the interaction between these indicators and household size indicators), and pretreatment covariates --- and not treatment assignment --- as predictors.\footnote{The compliance ensemble is evaluated in terms of 10--fold cross--validated MSE. The distribution of MSE for the ensemble and its candidate algorithms are provided in Table ~\ref{compliance-ensemble}.}  This gives evidence in favor of the conditional independence assumption.

\subsubsection{Strong ignorability} We cannot directly test Assumptions \eqref{si_treat} and \eqref{si_ctrl}, which state that potential outcomes for treatment and control are independent of sample assignment for individuals with the same covariates and assignment to treatment. The assumptions are only met if every possible confounder associated with the response and the sample assignment is accounted for. In estimating the response surface, we use all demographic, socioeconomic, and pre-existing health condition data that were common in the OHIE and NHIS data. Potentially important unobserved confounders include the number of hospital and outpatient visits in the previous year, proximity to health services, and enrollment in other federal programs. 

The final two columns of Table~\ref{rct-nrt-compare} compares RCT participants selected for Medicaid with population members on Medicaid. Compared to the RCT compliers, the population members who received treatment are younger, female, and more racially and ethnically diverse. Diagnoses of diabetes, asthma, high blood pressure, and heart disease are more common among the population on Medicaid then the RCT treated.

Strong ignorability assumptions may also be violated due to the fact that the OHIE applied a more stringent exclusion criteria compared to the NHIS sample. While the RCT and population sample both screened for individuals below the FPL, only the RCT required those enrolled to recertify their household income eligibility during the study period. Strong ignorability would not hold if the failure to recertify is correlated with unobserved variables.

\subsubsection{No defiers} \label{sens-defiers}

Angrist et al. \cite{Angrist1996} show that the bias due to violations of Assumption \eqref{monotonicity} is equivalent to the difference of average causal effects of treatment received for compliers and defiers, multiplied by the relative proportion of defiers, 
$\pr(i\text{ is a defier}) / (\pr(i\text{ is a complier]}) - \pr(i\text{ is a defier})).$

Table~\ref{ohie-status} reports the distribution of participants in the OHIE by status of treatment assignment and treatment received. Assumption \eqref{monotonicity} does not hold due to the presence of defiers; i.e., participants who were assigned to control and enrolled in Medicaid during the study period. About 7\% of the RCT sample were assigned to control but were enrolled in Medicaid ($T_i < D_i$) and 66\% of the sample complied with treatment assignment ($D_i = T_i$), which results in a bias multiplier of 0.1. Suppose that the difference of average causal effects of Medicaid received on ER use for compliers and defiers is 1.2\%. The resulting bias is only 0.1\%, which would not meaningfully alter the interpretation of the PATT-C or CACE estimates. 

\subsubsection{Implied assumptions} We implicitly assume no interference between households in the OHIE because treatment assignment occured at the household level. Within-household interference is not possible in this RCT because household members share the same treatment status. Interference between households would threaten the no-interference assumption in the unlikely case that the Medicaid coverage of individuals in treated households affects the health care use of individuals in households assigned to control.

The implied exclusion restriction assumption ensures treatment assignment affects the response only through enrollment in Medicaid. It is reasonable that a person's enrollment in Medicaid, not just their eligibility to enroll, would affect their hospital use. For private health insurance one might argue that eligibility may be be negatively correlated with hospital use, as people with pre-existing conditions are less often eligible yet go to the hospital more frequently. This should not be the case with a federally funded program such as Medicaid. 

\subsection{Placebo tests} \label{placebo-tests}

We conduct placebo tests to check whether the average outcomes differ between the RCT compliers on Medicaid and the adjusted population members who received Medicaid. If the placebo tests detect a significant difference between the mean outcomes of these groups, it would indicate bias arising from violations of the identification and modeling assumptions.

Table~\ref{placebo} reports the results of placebo tests for the weighted difference-in-mean outcomes of RCT compliers and adjusted population members who received treatment. The mean outcome of RCT compliers is calculated from the observed RCT sample and is weighted by OHIE survey weights. The adjusted population mean is the counterfactual $Y_{i11}$ estimated from \ref{response-predict} of the estimation procedure, and weighted by NHIS survey weights.

We first perform a Test of One-Sided Significance (TOST) \cite{berger1996bioequivalence} that evaluates equivalence between the weighted distributions, with a failure to reject the null of a substantively large difference. The TOST $p$-value is below the conventional level of significance ($p \leq 0.05$), indicating rejection of the null of a substantively large difference. Secondly, we conduct standard tests-of-difference and fail to reject the null of no difference. These results imply that the PATT-C estimator is not biased by differences in how Medicaid is delivered or health outcomes are measured between the RCT and population, or by differences in the unobserved characteristics of individuals in the sample or population. 
 
\subsection{Empirical results}\label{results}

Table~\ref{compliance-compare} compares population and sample estimates of the effect of Medicaid on health care use. The PATT-C estimates indicate that Medicaid coverage has a statistically significant and negative effect on the number of ER and outpatient visits. The PATT estimates, which do not account for noncompliance in the RCT, indicate a similarly sized negative effect on the number of ER visits and no effect on the number of outpatient visits attributable to Medicaid coverage. Our population estimates are consistent with the study of Kowalski \cite{NBERw22363}, which extrapolates LATE estimates to the Massachusetts population and find negative LATEs on ER utilization. 
  
The PATT-C estimates differ in direction and magnitude relative the CACE estimates on the OHIE sample, which can be explained by differences in the covariate distributions of the RCT sample and population. The CACE estimates indicate a positive and significant effect on primary care visits and no effect on ER use. The direction and magnitude of the CACE estimated treatment effect on ER use is almost identical to the corresponding LATE estimates on the RCT sample reported by Finkelstein et al. \cite{finkelstein2012}, although the authors uncover a significant (positive) effect of Medicaid on ER use.

\begin{table}
	\centering
\begin{threeparttable}[htbp]
	\caption{Comparison of population and sample estimates.\label{compliance-compare}} 
	\begin{tabular}{@{}lcc@{}}
		\toprule
		\backslashbox{Estimator}{Outcome} 				& $\#$ ER visits   & $\#$ outpatient visits      \\ 
		\midrule
		PATT-C                                          & -0.002 	& -0.017  \\
			                                            & (0.0006)	& (0.0008)  \\
		PATT                                            & -0.001 	& 0.003  \\
				                                        & (0.0004)	& (0.0027)  \\
		CACE                                            & 0.023  	& 0.225   \\ 
					                                    & (0.0327)  & (0.0920)   \\ \bottomrule
	\end{tabular}
  \begin{tablenotes}[flushleft]
	\small
	\item \emph{Notes:} survey-weighted difference-in-means; bootstrapped standard errors (in parentheses) are clustered on the household. Response models include binary indicators for treatment received, household size, survey wave (and interactions), and pretreatment covariates. 
\end{tablenotes}
\end{threeparttable}
\end{table}

Treatment effect heterogeneity in the population helps to explain the differences between the complier-adjusted sample and population treatment effect estimates. Figure~\ref{fig:het-plot} plots PATT-C estimates of heterogeneous treatment effects on ER and outpatient use in the population. Heterogeneous treatment effects are estimated by taking differences across response surfaces for a given binary covariate, and response surfaces are estimated with the ensemble mean predictions. We find significant decreases in ER use due to Medicaid coverage for adult population members who are female, black, below the age of 50, have a vocational or two-year degree, and have household income of over $\$10,000$. Population members with diabetes, asthma, and heart conditions are also less likely to visit the ER due to Medicaid coverage. The estimated effect of Medicaid on the number of outpatient visits is negative for population members who are over the age of 50, black, diagnosed with a heart condition, have a vocational or two-year degree, and have household income of over $\$10,000$.

\begin{figure}[htbp]
	\begin{center}
		\includegraphics[width = 0.49\textwidth]{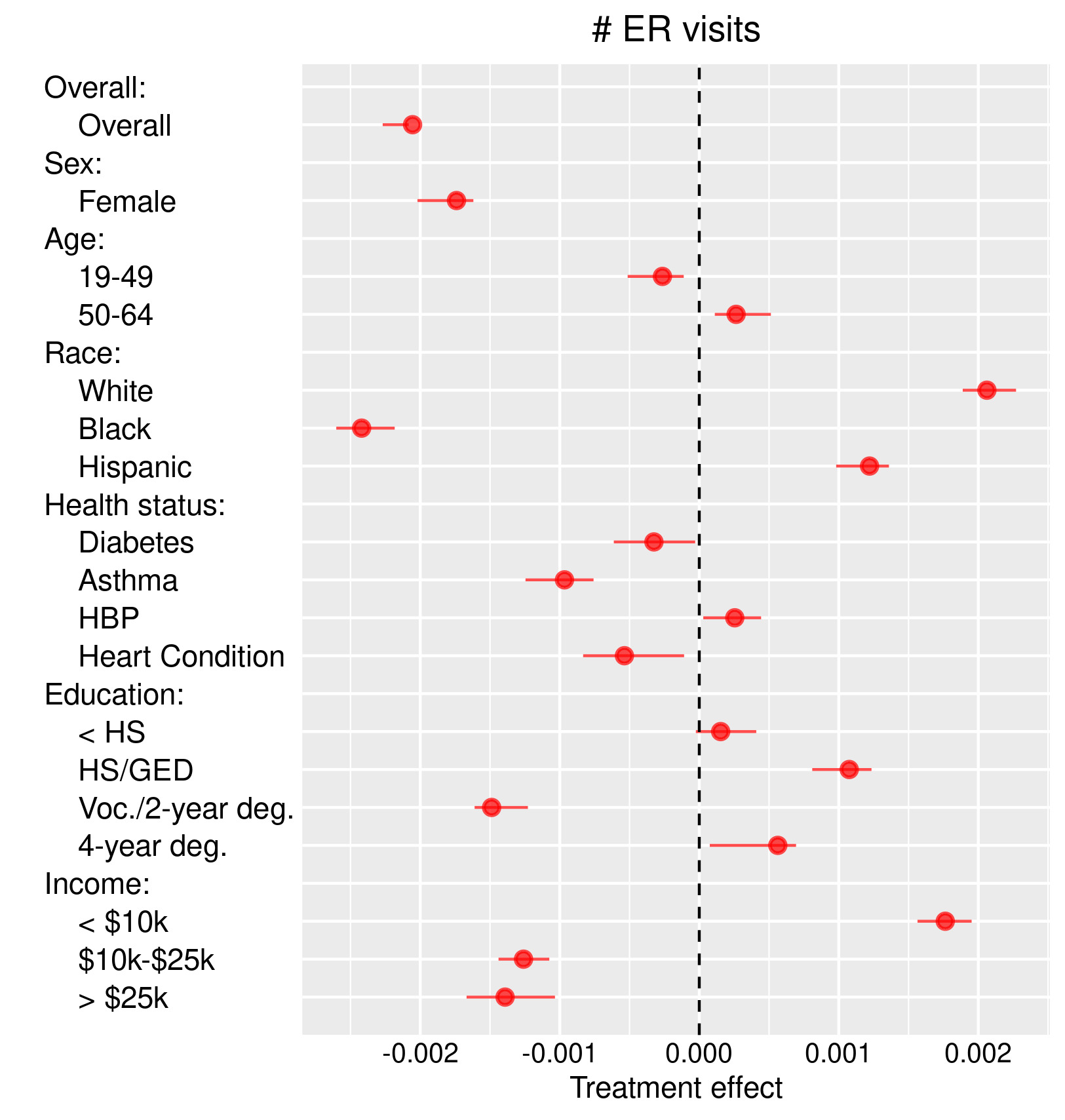} 
		\includegraphics[width = 0.49\textwidth]{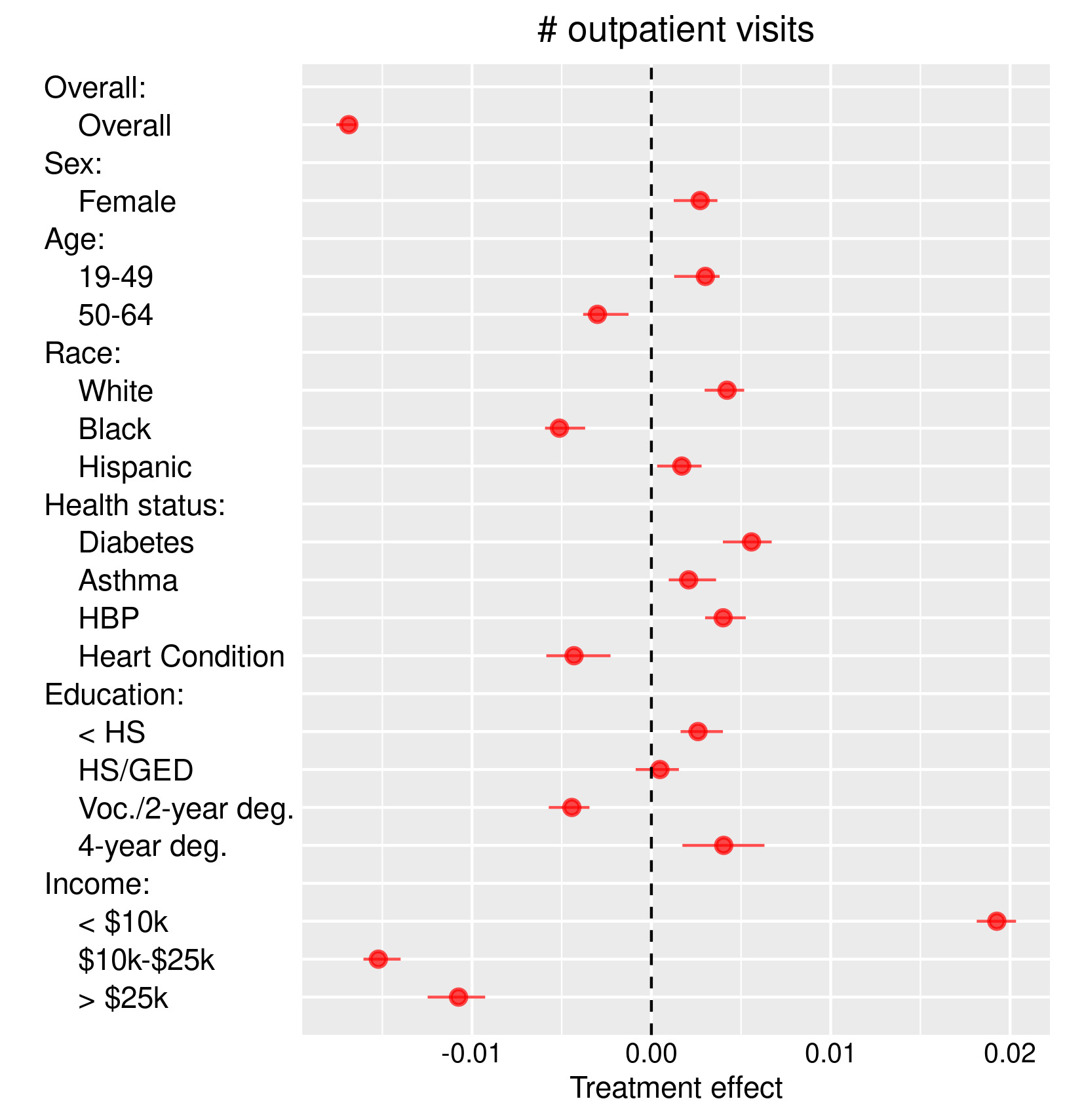}
		\caption{Heterogeneity in PATT-C treatment effect estimates. \emph{Notes:} Horizontal lines represent 95\% percentile bootstrap confidence intervals.\label{fig:het-plot}}
	\end{center}
\end{figure}

\section{Discussion} \label{discussion}

The simulation results presented in Section~\ref{sim} show that the PATT-C estimator outperforms its unadjusted counterpart when the population compliance rate is low. Of course, the simulation results depend on the particular way we parameterized the compliance, selection, treatment assignment, and response schedules. 

In particular, the strength of correlation between the covariates and compliance governs how well the estimator will perform, since \ref{compliance-model} of the estimation procedure predicts who \textit{would be} a complier in the RCT control group, had they been assigned to treatment. If it is difficult to predict compliance using the observed covariates, then the estimator will perform badly because of noise introduced by incorrectly treating noncompliers as compliers. Further research should be done into ways to test how well the model of compliance works in the population or explore models to more accurately predict compliance in RCTs.  Accurately predicting compliance is not only essential for yielding unbiased estimates of the average causal effects for target populations, it is also useful for researchers and policymakers to know which groups of individuals are unlikely to comply with treatment. 

In the OHIE trial, less than one-third of those selected to receive Medicaid benefits actually enrolled. We accurately classified compliance status for 78\% of treated RCT participants using only the pretreatment covariates as features. While we do not know how well the compliance ensemble predicts for the control group, the control group should be similar to the treatment group on pretreatment covariates because of the RCT randomization. The model's performance on the training set suggests that compliance is not purely random and depends on observed covariates, which gives evidence in favor of using PATT-C. 

In the empirical application, the sample population differs in several dimensions from the target population of individuals who would be covered by other Medicaid expansions, such as the ACA expansion. For instance, the RCT participants are disproportionately white and over the age of 50. The RCT participants volunteered for the study and therefore may be in poorer health compared to the target population. These differences in baseline covariates make reweighting or response surface methods necessary to extend the RCT results to the population.

Explicitly modeling compliance allows us to decompose population estimates by subgroup according to pretreatment covariates common to both RCT and observational datasets; e.g, demographic variables, pre-existing conditions, and insurance coverage. We find substantial treatment effect heterogeneity in terms of race, education, and health status subgroups. This pattern is expected because RCT participants volunteered for the
study. 

The PATT-C estimates indicate that Medicaid coverage significantly decreases the number of ER and outpatient visits. The estimated effect on ER utilization is consistent with quasi-experimental studies on the impact of the 2006 Massachusetts health reform and a previous study that extrapolates LATE estimates on the OHIE sample to the Massachusetts population. The PATT-C estimates differ in direction and magnitude relative the complier-adjusted sample estimates, which can be explained by treatment effect heterogeneity in the population and the differences in covariate distributions of the RCT sample and population.

\section*{Acknowledgments}
Poulos acknowledges support from the National Science Foundation Graduate Research Fellowship [grant number DGE 1106400]. Code for reproducing the results in this paper is available in the public repository \url{https://github.com/jvpoulos/patt-c}.

\pagebreak


\pagebreak
\begin{appendices}
	
\newcommand{\hbAppendixPrefix}{A}
\renewcommand{\thefigure}{\hbAppendixPrefix\arabic{figure}}
\setcounter{figure}{0}
\renewcommand{\thetable}{\hbAppendixPrefix\arabic{table}} 
\setcounter{table}{0}
\renewcommand{\theequation}{\hbAppendixPrefix\arabic{equation}} 
\setcounter{equation}{0}

\section{Proof of Theorem \eqref{thm1}}\label{appendix-proof}

\begin{proof}
	We separate the expectation linearly into two terms and consider each individually.
	\begin{align*}
	\ex\left(Y_{i01} \mid S_i=0,D_i=1\right) &= \ex\left(Y_{i11} \mid S_i=0, D_i=1\right) \tag*{by Assumption \eqref{consistency}} \\
	&= \ex\left(Y_{i11} \mid S_i=0, T_i=1, C_i=1\right) \tag*{by Assumption \eqref{monotonicity}} \\
	&= \ex_{01}\left[  \ex\left(Y_{i11} \mid S_i=0, T_i=1, C_i=1, W_i\right) \right] \\
	&= \ex_{01}\left[  \ex\left(Y_{i11} \mid S_i=1, T_i=1, C_i=1, W_i\right) \right] \tag*{by Assumption \eqref{si_treat}} \\
	&= \ex_{01}\left[  \ex\left(Y_{i11} \mid S_i=1, D_i=1, W_i\right)\right]. 
	\end{align*}
	Intuitively, conditioning on $W_i$ makes sample selection ignorable under Assumption \eqref{si_treat}. This is the critical connector between the third and fourth lines of the first expectation derivation.
	
	\begin{align*}
	\ex\left(Y_{i00} \mid S_i=0, D_i=1\right) &= \ex\left(Y_{i10} \mid S_i=0, D_i=1\right) \tag*{by Assumption \eqref{consistency}} \\
	&= \ex\left(Y_{i10} \mid S_i=0, T_i=1, C_i=1\right) \tag*{by Assumption \eqref{monotonicity}} \\
	&= \ex_{01}\left[  \ex\left(Y_{i10} \mid S_i=1, T_i=1, C_i=1, W_i\right) \right] \tag*{by Assumption \eqref{si_ctrl}} \\
	&= \ex_{01}\left[  \ex\left(Y_{i10} \mid S_i=1, T_i=0, C_i=1, W_i\right) \right] \tag*{by Assumption \eqref{compl}}.
	\end{align*}
	The last line follows because Assumption \eqref{compl} allows us to use RCT controls who would have complied had they been assigned to treatment. Finally, the result follows by plugging these two expressions into Eq.~\eqref{tpattc}.
\end{proof}

\section{Tables \& Figures}

\begin{figure}[htbp]
	\begin{center}
		\includegraphics[width = 1\textwidth]{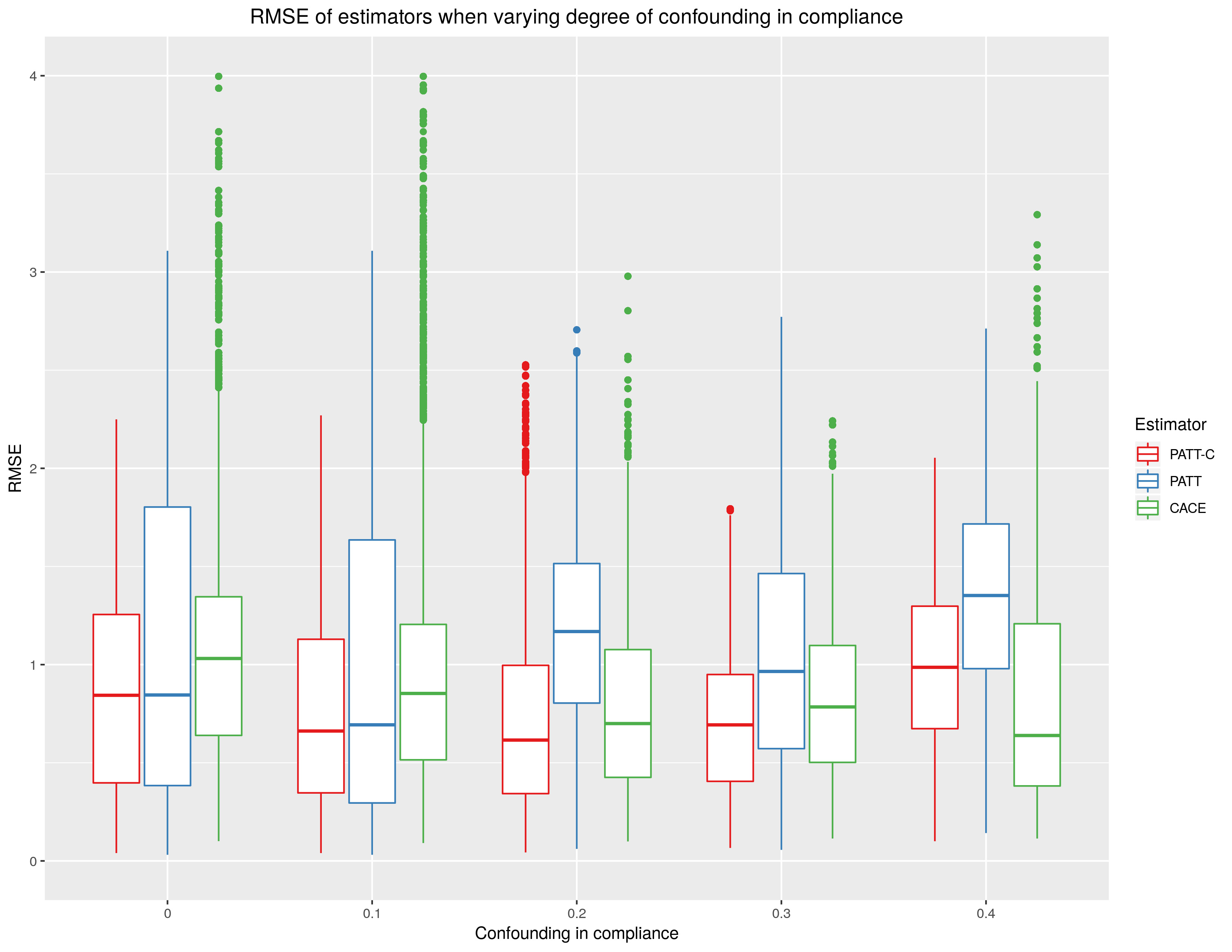}
		\caption{Average RMSE according to degree of confounding in compliance.\label{fig:rmse_boxplots_RateConC}}
	\end{center}
\end{figure}

\begin{figure}[htbp]
	\begin{center}
		\includegraphics[width = 1\textwidth]{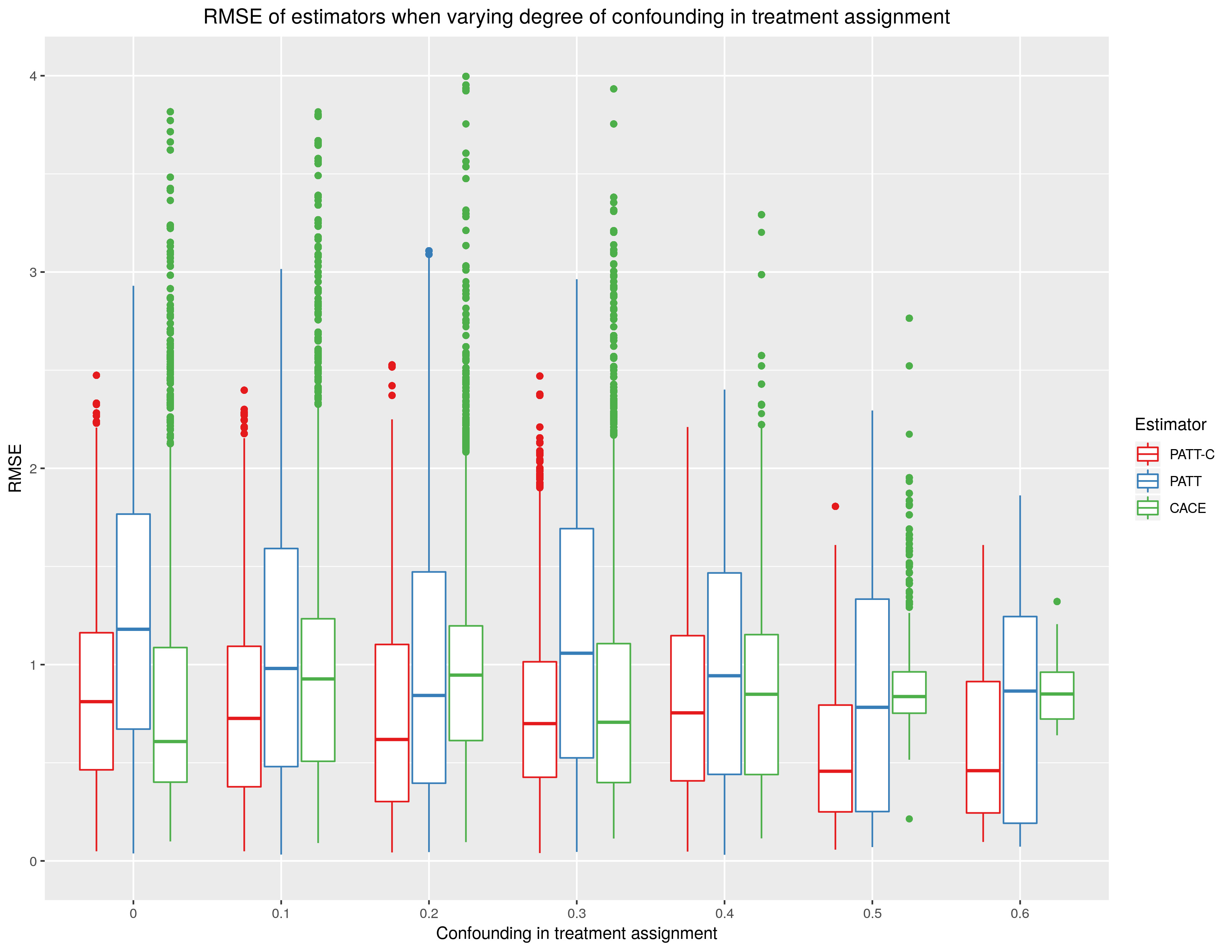}
		\caption{Average RMSE according to degree of confounding in treatment assignment.\label{fig:rmse_boxplots_RateConT}}
	\end{center}
\end{figure}

\begin{figure}[htbp]
	\begin{center}
		\includegraphics[width = 1\textwidth]{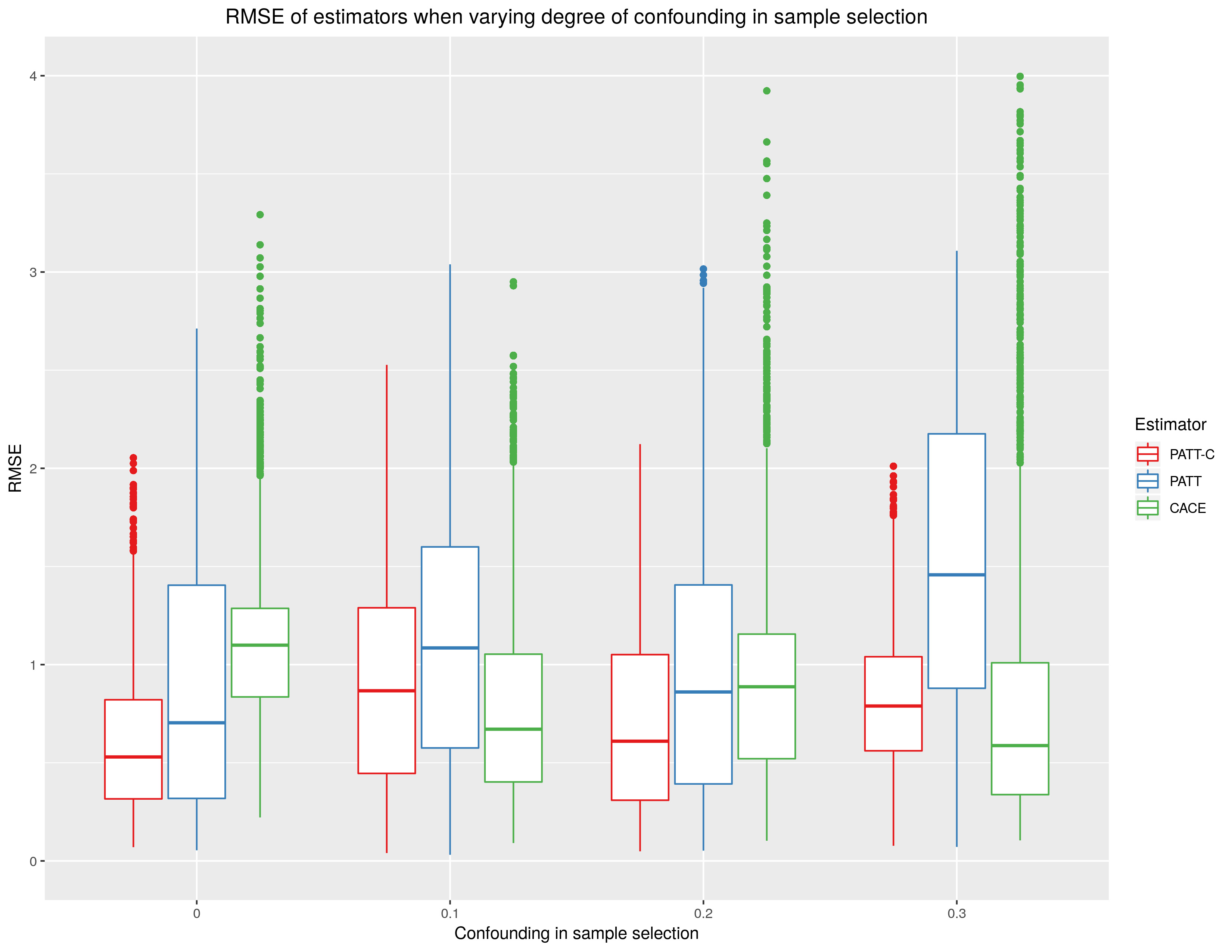}
		\caption{Average RMSE according to degree of confounding in sample selection.\label{fig:rmse_boxplots_RateConS}}
	\end{center}
\end{figure}

\begin{table}[htbp]
	\centering
	\begin{threeparttable}
		\caption{Distribution of MSE for compliance ensemble.\label{compliance-ensemble}} 
		\begin{tabular}{lcccc}
			\hline
			Algorithm & Mean & SE & Min. & Max. \\ 
			\hline
			Super learner (\texttt{SuperLearner}) & 0.22  & 0.001 & 0.21 & 0.23 \\
			\hline
			Lasso regression (\texttt{glmnet})  & 0.22  & 0.001 & 0.21 & 0.23 \\
			Random forests (\texttt{randomForest}) & 0.27  & 0.002 & 0.25 & 0.29 \\
			Ridge regression (\texttt{glmnet}) & 0.22  & 0.001 & 0.21 & 0.23 \\
			\hline
		\end{tabular}
	  \begin{tablenotes}[flushleft]
		\small
	\item MSE is 10-fold cross-validated error for super learner ensemble and candidate algorithms. \textsf{R} package used for implementing each algorithm in parentheses.
\end{tablenotes}

\end{threeparttable}
\end{table}

\begin{singlespace}
	\centering
			\begin{threeparttable}
	\begin{longtable}{lllllll}
		\caption{Pretreatment covariates and responses for OHIE and NHIS respondents by Medicaid coverage status.\label{rct-nrt-compare}} \\
		& OHIE &  & OHIE &  & NHIS &  \\ 
		& no Medicaid &  & Medicaid &  &Medicaid &   \\ 
		& $n=13,573$ &  & $n=5,547$ &  & $n=6,256$ &  \\  
		\hline   
		\hline   
		\textbf{Covariate} &  N & $\mathbf{\%}$ & N & $\mathbf{\%}$ & N & $\mathbf{\%}$ \\ 
		\hline
		\textit{Age:} &  & & &  &  & \\ 
		\hspace{3mm}19-49 & 4,166 & 31.0 & 1,473 & 26.9 & 4,322 & 69.2  \\ 
		
		\hspace{3mm}50-64 & 9,407 & 69.0 & 4,074 & 73.1 & 1,934 & 30.8 \\
			&  & & &  &  & \\ 

		\textit{Education:} &  & & &  &  & \\  
		\hspace{3mm}Less than high school  & 2,181 & 15.8 & 1,028 & 18.3 & 1,941 & 30.8  \\ 
		
		\hspace{3mm}High school diploma or GED & 6,948 & 50.9 & 2,948 & 52.7 & 2,074 & 33.3   \\ 
		
		\hspace{3mm}Voc. training / 2-year degree & 2,889 & 21.5 & 1,118 & 20.8 & 1,809 & 29.1  \\ 
		
		\hspace{3mm}4-year college degree or more & 1,555 & 11.8 & 453 & 8.2 & 432 & 6.8  \\ 
		&  & & &  &  & \\ 

		\textit{Ethnicity:} &  & & &  &  & \\ 

		\hspace{3mm}Black & 193 & 3.1 & 247 & 3.5 & 1,723 & 27.0  \\ 
		
		\hspace{3mm}Hispanic &  1,242 & 9.3 & 509 & 9.0 & 1,570 & 24.9  \\ 
				
		\hspace{3mm}White & 11,453 & 84.2 & 4,675 & 84.7 & 3,899 & 63.0  \\ 
		
		&  & & &  &  & \\ 
		\textit{Gender:} &  & & &  &  & \\ 
		\hspace{3mm} Female & 7,807 & 57.8 & 3,280 & 59.3 & 4,285 & 68.4 \\ 
		&  & & &  &  & \\ 
		\textit{Health status:} &  & & &  &  & \\ 
		
		\hspace{3mm}Asthma & 2,114 & 15.6 & 1,004 & 18.3 & 1,269 & 20.1   \\ 
		
		\hspace{3mm}Diabetes & 1,484 & 10.7 & 595 & 11.0 & 864 & 14.0  \\ 
				
		\hspace{3mm}Heart condition & 302 & 2.1 & 178 & 3.1 & 529 & 8.5 \\ 
		
		\hspace{3mm}High blood pressure (HBP) & 3,740 & 27.4 & 1,533 & 28.1 & 2,162 & 34.9  \\ 	
		&  & & &  &  & \\ 
		\textit{Income:} &  & & &  &  & \\ 
		\hspace{3mm} $<\$10$k & 6,003 & 43.7 & 3,488 & 62.6 & 2,587 & 41.8  \\
		
		\hspace{3mm} \$10k-\$25k & 5,583  & 41.7 & 1,678 & 30.5 & 3,095 & 49.2 \\
		
		\hspace{3mm} $>\$25$k & 1,987 & 14.6 & 381 & 6.9 & 574 & 9.0   \\		
			
		\hline
		\hline
		\textbf{Responses} &  Mean & S.d. & Mean & S.d. & Mean & S.d. \\  
		\hspace{3mm}$\#$ ER visits &  0.45 & 1.00 & 0.40 & 0.85 & 0.25 & 0.54  \\  
		\hspace{3mm}$\#$ outpatient visits & 1.98 & 2.97 & 1.98 & 2.92 & 1.07 & 1.21 \\
		\hline
		\hline
	\end{longtable}
  \begin{tablenotes}[flushleft]
	\small
	\item \emph{Notes:} percentages calculated using OHIE or NHIS survey weights.
\end{tablenotes}
\end{threeparttable}
\end{singlespace}

\begin{table}[htbp]
	\centering
	\begin{threeparttable}
		\caption{Placebo tests.\label{placebo}} 
	\begin{tabular}{@{}lccccc@{}}
		\toprule
		Outcome                & RCT complier & Adjusted population & Weighted & TOST & Two-sided\\
				               &   weighted mean 		  & weighted mean 	& 	difference-in-means	   &  $p$-value & $p$-value \\		
		 \midrule
		$\#$ ER vists          & 0.449              & 0.466                           & -0.017      & $< 0.001$ & 0.202 \\	 
		$\#$ outpatient visits & 1.966              & 2.029                           & -0.063      & $< 0.001$ & 0.090 \\ \bottomrule
	\end{tabular}
  \begin{tablenotes}[flushleft]
	\small
\item \emph{Notes:} $p$-values correspond to survey-weighted difference-in-means estimated from bootstrapped standard errors, which are clustered on the household. TOST $p$-value evaluates equivalence between the weighted distributions, and the two-sided $p$-value evaluates difference between the weighted distribution.
\end{tablenotes}

\end{threeparttable}
\end{table}

\begin{table}[htbp]
	\centering
	\begin{threeparttable}
	\caption{Distribution of OHIE participants by status of treatment assignment ($T_i$) and treatment received ($D_i$).\label{ohie-status}} 
	\begin{tabular}{@{}lccc@{}}
		\toprule
		& $D_i = 0$ & $D_i = 1$ & N      \\ \midrule
		$T_i = 0$ & 8,343    & 1,265     & 9,608 \\
		$T_i = 1$ & 5,230     & 4,282     & 9,512 \\
		N         & 13,573    & 5,547     & 19,120 \\ \bottomrule
	\end{tabular}
\end{threeparttable}
\end{table}

\begin{table}[htbp]
\centering
\begin{threeparttable}
\caption{Error and weights for candidate algorithms in response ensemble for RCT compliers.\label{reponse-ensemble}}  
  \begin{tabularx}{\linewidth}{l*{3}{Y}}
    \toprule
    \multicolumn{3}{l}{\textbf{$\#$ ER visits}} \\
    \midrule
Algorithm  & MSE & Weight \\ 
\hline
Additive regression, $\text{degree} = 2$ (\texttt{gam})  & 0.81 & 0 \\
Additive regression, $\text{degree} = 3$ (\texttt{gam})  & 0.81 & 0 \\ 
Additive regression, $\text{degree} = 4$ (\texttt{gam})  & 0.81 & 0 \\ 
Boosted regression (\texttt{gbm})  & 0.79 & 0 \\
Lasso regression (\texttt{glmnet})  & 0.79 & 0.76 \\ 
Random forests, $\# preds. = 1$ (\texttt{randomForest}) & 0.79 & 0.18 \\ 
Random forests, $\# preds. = 10$ (\texttt{randomForest}) & 0.83 & 0.06 \\ 
Regularized linear regression, $\alpha=0.25$ (\texttt{glmnet})  & 0.79 & 0 \\ 
Regularized linear regression, $\alpha=0.5$ (\texttt{glmnet})  & 0.79 & 0 \\ 
Regularized linear regression, $\alpha=0.75$ (\texttt{glmnet})  & 0.79 & 0 \\ 
Ridge regression (\texttt{glmnet})  & 0.79 & 0 \\ 
  \end{tabularx}
  \begin{tabularx}{\linewidth}{l*{3}{Y}}
	\toprule
	\multicolumn{3}{l}{\textbf{$\#$ outpatient visits}} \\
	\midrule
	Algorithm  & MSE & Weight \\ 
	\hline
	Additive regression, $\text{degree} = 2$ (\texttt{gam})  & 6.60 & 0 \\
	Additive regression, $\text{degree} = 3$ (\texttt{gam})  & 6.60 & 0 \\ 
	Additive regression, $\text{degree} = 4$ (\texttt{gam})  & 6.60 & 0 \\ 
	Boosted regression (\texttt{gbm})  & 6.48 & 0 \\
	Lasso regression (\texttt{glmnet})  & 6.47 & 0.87 \\ 
	Random forests, $\# preds. = 1$ (\texttt{randomForest}) & 6.47 & 0 \\ 
	Random forests, $\# preds. = 10$ (\texttt{randomForest}) & 6.73 & 0.13 \\ 
	Regularized linear regression, $\alpha=0.25$ (\texttt{glmnet})  & 6.47 & 0 \\ 
	Regularized linear regression, $\alpha=0.5$ (\texttt{glmnet})  & 6.47 & 0 \\ 
	Regularized linear regression, $\alpha=0.75$ (\texttt{glmnet})  & 6.47 & 0 \\ 
	Ridge regression (\texttt{glmnet})  & 6.47 & 0 \\ 
	\hline
	\bottomrule
\end{tabularx}
  \begin{tablenotes}[flushleft]
	\small
\item \emph{Notes:} cross-validated error and weights used for each algorithm in super learner ensemble. \textit{MSE} is the ten-fold cross-validated mean squared error for each algorithm. \textit{Weight} is the coefficient for super learner, which is estimated using non-negative least squares based on the Lawson-Hanson algorithm. \textsf{R} package used for implementing each algorithm in parentheses. $\# preds.$ is the number of predictors randomly sampled as candidates in each decision tree in random forests algorithm. $\alpha$ is a parameter that mixes L1 and L2 norms. $\text{degree}$ is the smoothing term for smoothing splines.
\end{tablenotes}

\end{threeparttable}
\end{table}

\pagebreak

\begin{table}[htbp]
\centering
\begin{threeparttable}
\caption{Error and weights for candidate algorithms in response ensemble for all RCT participants.\label{reponse-ensemble-unadj}}  
\begin{tabularx}{\linewidth}{l*{3}{Y}}
	\toprule
	\multicolumn{3}{l}{\textbf{$\#$ ER visits}} \\
	\midrule
	Algorithm  & MSE & Weight \\ 
	\hline
	Additive regression, $\text{degree} = 2$ (\texttt{gam})  & 0.76 & 0 \\ 
	Additive regression, $\text{degree} = 3$ (\texttt{gam})  & 0.76 & 0 \\ 
	Additive regression, $\text{degree} = 4$ (\texttt{gam})  & 0.76 & 0 \\ 
	Boosted regression (\texttt{gbm})  & 0.75 & 0 \\
	Lasso regression (\texttt{glmnet})  & 0.75 & 0 \\ 
	Random forests, $\# preds. = 1$ (\texttt{randomForest}) & 0.75 & 1 \\ 
	Random forests, $\# preds. = 10$ (\texttt{randomForest}) & 0.79 & 0 \\ 
	Regularized linear regression, $\alpha=0.25$ (\texttt{glmnet})  & 0.75 & 0 \\ 
	Regularized linear regression, $\alpha=0.5$ (\texttt{glmnet})  & 0.75 & 0 \\ 
	Regularized linear regression, $\alpha=0.75$ (\texttt{glmnet})  & 0.75 & 0 \\ 
	Ridge regression (\texttt{glmnet})  & 0.75 & 0 \\ 
\end{tabularx}
\begin{tabularx}{\linewidth}{l*{3}{Y}}
	\toprule
	\multicolumn{3}{l}{\textbf{$\#$ outpatient visits}} \\
	\midrule
	Algorithm  & MSE & Weight \\ 
	\hline
	Additive regression, $\text{degree} = 2$ (\texttt{gam})  & 6.44 & 0 \\ 
	Additive regression, $\text{degree} = 3$ (\texttt{gam})  & 6.44 & 0 \\ 
	Additive regression, $\text{degree} = 4$ (\texttt{gam})  & 6.44 & 0 \\ 
	Boosted regression (\texttt{gbm})  & 6.39 & 0 \\
	Lasso regression (\texttt{glmnet})  & 6.39 & 0.87 \\ 
	Random forests, $\# preds. = 1$ (\texttt{randomForest}) & 6.39 & 0 \\ 
	Random forests, $\# preds. = 10$ (\texttt{randomForest}) & 6.61 & 0.13 \\ 
	Regularized linear regression, $\alpha=0.25$ (\texttt{glmnet})  & 6.39 & 0 \\ 
	Regularized linear regression, $\alpha=0.5$ (\texttt{glmnet})  & 6.39 & 0 \\ 
	Regularized linear regression, $\alpha=0.75$ (\texttt{glmnet})  & 6.39 & 0 \\ 
	Ridge regression (\texttt{glmnet})  & 6.39 & 0 \\ 
	\hline
	\bottomrule
\end{tabularx}
  \begin{tablenotes}[flushleft]
	\small
	\item See notes to Table~\ref{reponse-ensemble}.
\end{tablenotes}
\end{threeparttable}
\end{table}

\end{appendices}

\itemize
\end{document}